\documentclass{article}

\usepackage{gastex}
\usepackage{color}
\usepackage{amsmath}
\usepackage{amssymb}
\usepackage{theorem}
\newtheorem{theorem}{Theorem}[section]

\newtheorem{proposition}[theorem]{Proposition}
\newtheorem{corollary}{Corollary}
\newtheorem{lemma}{Lemma}
{\theorembodyfont{\rmfamily}%
 
  \newtheorem{example}[theorem]{\sc Example}
  
}
\newenvironment{proof}{\noindent\textit{Proof}}
                      {\QED\vskip\theorempostskipamount}
\def\petitcarre{\vrule height4pt width 4pt depth0pt}
\def\endexamplesymbol{$\bullet$}
\def\endex{\relax\ifmmode\eqno{\hbox{\endexamplesymbol}}\else{%
  \unskip\nobreak\hfil\penalty50\hskip2em\hbox{}\nobreak\hfil
\newtheorem{lemma}{Lemma}  \endexamplesymbol
  \parfillskip=0pt \finalhyphendemerits=0\par\smallskip}
  \fi}
\def\QED{\relax\ifmmode\eqno{\hbox{\petitcarre}}\else{%
 \unskip\nobreak\hfil\penalty50\hskip2em\hbox{}\nobreak\hfil
  \petitcarre
  \parfillskip=0pt \finalhyphendemerits=0\par\smallskip}
 \fi}

\def\petitcarre{\vrule height4pt width 4pt depth0pt}
\def\QED{\relax\ifmmode\eqno{\hbox{\petitcarre}}\else{%
  \unskip\nobreak\hfil\penalty50\hskip2em\hbox{}\nobreak\hfil
  \petitcarre
  \parfillskip=0pt \finalhyphendemerits=0\par\smallskip}
  \fi}
\def\Z{\mathbb Z}
\def\C{\mathbb C}
\def\A{\mathcal A}
\def\B{\mathcal B}
\def\V{\mathcal V}
\def\Agoth{{\mathfrak A}}
\newcommand{\GR}{{\cal R}}
\newcommand{\GL}{{\cal L}}

\newcommand{\GH}{{\cal H}}
\newcommand{\GD}{{\cal D}}
\newcommand{\edge}[1]{\stackrel{#1}{\rightarrow}}
\DeclareMathOperator{\End}{End}
\DeclareMathOperator{\Trace}{Tr}
\DeclareMathOperator{\Card}{Card}
\def\u(#1){\underline{#1\!}\,}

\title{Completely reducible sets}
\author{Dominique Perrin\\ LIGM, Universit\'e Paris-Est}
\begin{document}
\maketitle
\begin{abstract}
We study the family of rational sets of words, called completely
reducible and
which are such that the syntactic
representation of their characteristic series is completely
reducible. This family contains,
by a result of Reutenauer, the submonoids generated by bifix
codes and, by a result of Berstel and Reutenauer,
the cyclic sets. We study the closure properties
 of this family. We prove a result
on linear representations of monoids which gives a generalization of
the result concerning the complete reducibility
of the submonoid generated by a bifix code to sets called birecurrent.
 We also give a new proof of the result
concerning cyclic sets.

\end{abstract}
\tableofcontents

\section{Introduction}
The notion of syntactic algebra of a formal series was
introduced by Reutenauer  in~\cite{Reutenauer1980}.
  It is a natural
generalization of the notion of syntactic monoid of a set of words.
This algebra has a natural linear representation called the syntactic
representation of the series, in the same way as the syntactic
monoid has a natural representation by mappings from a set
into itself corresponding to the minimal automaton of the set.

 In the same
way that one uses properties of the syntactic monoid of a set
to define or characterize important classes of sets, it is natural to use
the syntactic algebra to do the same. One of the most elementary
property of a linear
representation is its irreducibility or, more interestingly, its
complete reducibility. The syntactic representation and the syntactic
algebra of a set of words
is those of its characteristic series.
A set of words is called completely reducible
if its syntactic representation  is
completely reducible. This is equivalent to the semisimplicity
of its syntactic algebra.

A remarkable property, also proved by Reutenauer in \cite{Reutenauer1981}
is that, when the field is of characteristic zero, 
the submonoid generated by a rational bifix code
is completely reducible. This can be considered as a generalization
of Maschke's theorem and is one of the arguments showing the strong
connexion between bifix codes and groups. Later, Berstel and Reutenauer proved 
in~\cite{BerstelReutenauer1990} that the sets of words
called cyclic,
are also completely reducible. The proofs given for both cases do not
have much in common. The proof given in~\cite{Reutenauer1981}
for the first result consists in proving that the radical
of the syntactic algebra of the set is zero. Another proof, 
given in \cite{BerstelPerrin1985} 
and in \cite{BerstelPerrinReutenauer2009}, shows  directly that
the syntactic representation of the set
is completely reducible.
This is also the proof presented in~\cite{BerstelReutenauer2011}. The proof of
the other result on cyclic languages uses a decomposition
of the characteristic series as a $\Z$-linear combination
of series also called cyclic.

In this paper, we investigate further the family of completely
reducible sets. We study the closure properties of this family
(Theorem~\ref{theoremVariety}) and prove some necessary and some
sufficient conditions to belong to the family. We characterize
the completely reducible sets on a one letter alphabet (Theorem~\ref{theorem1letter}).
Reworking the proof of the complete reducibility
of submonoids generated by bifix codes,
we prove a result
on linear representations which gives a sufficient
condition for complete reducibility
(Theorem~\ref{theoremCompleteReducibility}).
It is related with the idea of condensation in representation theory,
originally  due to Green \cite{Green2007}, and used in computational
group theory.
We use Theorem~\ref{theoremCompleteReducibility} to
obtain a generalization of the complete reducibility of the
submonoid generated by a bifix code to a class of sets called
birecurrent (Theorem~\ref{theoremBifix}).
 They are defined by the property that the minimal
automata of the set and of its reversal are strongly connected.
We give a proof of
the
complete reducibility of cyclic sets which  uses the notion
of external power of an automaton introduced by B\'eal
in~\cite{Beal1995}, the results on strongly cyclic sets proved
in~\cite{BealCartonReutenauer1996} and the family of series defined
by traces used in the original proof of~\cite{BerstelReutenauer1990}.
We finally relate cyclic sets and monoid characters, based on the
results of McAlistair~\cite{McAlister1972}.

The problem of  characterizing  the completely reducible sets
in terms of operations on sets of words
remains open. It is solved on a one-letter alphabet
and for the class of submonoids generated by
rational maximal codes, since in this case the complete reducibility can only occur
for a bifix code by the result of Reutenauer already mentioned.
Such a characterization should take in account
the characteristic of the field since the complete reducibility
of the submonoid generated by a bifix code is only true when
the characteristic of the field is zero (or more generally
does not divide the order of the group of the bifix code).

The paper is organized as follows. In
Section~\ref{sectionMonoidsMatrices}
we prove the result concerning completely reducible
linear representations
(Theorem~\ref{theoremCompleteReducibility}).
In Section~\ref{SectionSyntactic} we define syntactic representations
and recall some results concerning them. In
Section~\ref{sectionCompletelyReducible}, we prove some closure properties
for the family of completely reducible sets.
We characterize this family on a one-letter alphabet 
(Theorem~\ref{theorem1letter}). In
Section~\ref{sectionBirecurrent},
we give a proof of the complete reducibility of birecurrent sets
(Theorem~\ref{theoremBifix}).
In Section~\ref{sectionCyclic} we give a new proof of the complete
reducibility
of cyclic sets (Theorem~\ref{theoremCyclic}).  We also describe the
connexion with a result of McAlistair on monoid characters~\cite{McAlister1972}.

\paragraph{Acknowledgements} I would like to thank Clelia De Felice for
her active and very important help in the preparation of this article.
Without her, the paper would not exist in this form. I also thank
Benjamin Steinberg for pointing out to me useful references on linear
representations of semigroups. Next, I am indebted to Christophe
Reutenauer for several improvements of the presentation and the
correction
of some mistakes in a preliminary version. Finally, special acknowledgments are
due
to the anonymous referee. He has contributed many improvements of this
paper
by pointing out the connexion with the condensation method,
with the Rhodes radical and also with a result of McAlistair 
on monoid characters.

\section{Completely reducible monoids of
  matrices}\label{sectionMonoidsMatrices}
In the first two parts of this section, we introduce basic notions concerning monoids and
linear representations. For a more detailed exposition, see~\cite{CliffordPreston1961}
or \cite{Lallement1979}. In the last part, we prove a result on
completely reducible representations (Theorem~\ref{theoremCompleteReducibility}) which will
be used in Section~\ref{sectionBirecurrent}.
\subsection{Monoids}

A semigroup is a set with an associative operation.
A monoid is a semigroup with an identity
element
denoted $1$. 

An element $0$ of a monoid $M$ is a \emph{zero} if $m\ne 1$
and for all $m\in M$,
$0m=m0=0$. If $M$ contains a zero, it is unique.

The Green relations on a monoid $M$ are defined as follows. For
$m,n\in M$, one has 
\begin{enumerate}
\item[(i)] $m\GR n$ if $mM=nM$.
\item [(ii)] $m\GL n$ if $Mm=Mn$.
\item[(iii)] $m\GH n$ if $m\GR n$ and $m\GL n$.
\end{enumerate}
It is classical that $\GR$ and $\GL$ commute. One denotes $\GD$ the
relation $\GR\GL=\GL\GR$. The $\GR$-class of $m\in M$ is denoted
$R(m)$
and similarly for $\GL$, $\GH$ and $\GD$.

A $\GD$-class $D$ is \emph{regular} if it contains an idempotent. In this
case,
there is an idempotent in each $\GR$-class and in each $\GL$-class of
$D$.

For any $m,n\in M$ one has $mn\in R(m)\cap L(n)$ if and only if
$R(n)\cap L(m)$ contains an idempotent (Clifford-Miller Lemma).
As a consequence, for any $m,m'$ in the same $\GH$-class $H$,
either $mm'\notin H$ or $H$ is a group.

A \emph{right ideal} (resp. a left ideal, resp. a two-sided ideal)
of a monoid $M$ is a nonempty subset $I$ such that
$IM=I$ (resp. $MI=I$, resp. $MIM=I$). A right ideal
 is \emph{minimal} if it does not contain
any other right ideal of $M$ (and similarly for left
and for two-sided ideals). In a finite monoid, there
is a unique minimal two-sided ideal which is the union of minimal
right ideals (resp. of minimal left ideals). 
When $M$ contains a zero, a right ideal
$I\ne 0$ is $0$-minimal if the only right ideals contained in $I$ are
$0$ and $I$ itself (and similarly for left and two-sided ideals).

\subsection{Linear representations of monoids}

Let $V$ be a vector space over a field $K$ and let $M$ be a submonoid
of the monoid $\End(V)$ of linear functions from $V$ into itself.
A subspace $V'$ of $V$ is \emph{invariant} by $M$ if $V'm\subset V'$
for any $m\in M$.
The monoid $M$ is called \emph{irreducible} if $V\ne 0$ and the only
invariant subspaces are $0$ and $V$.  
Otherwise, $M$ is called \emph{reducible}.

The monoid $M$ is \emph{completely
  reducible} if any invariant subspace has an invariant complement,
i.e. if for any invariant subspace $V'$ of $V$ there is an invariant
subspace $V''$ such that $V$ is the direct sum of $V'$ and $V''$.
If $V$ has finite dimension, a completely reducible submonoid
of $\End(V)$ has the following form.  There exists a decomposition of 
$V$ into a direct sum of invariant subspaces $V_{1}, V_{2},\ldots, V_{k}$,
\begin{displaymath}
  V = V_{1} \oplus V_{2} \oplus \cdots \oplus V_{k}
\end{displaymath}
such that the restrictions of the elements of $M$ to each of the
$V_{i}$'s form an irreducible submonoid of $\End(V_{i})$. 
Any invariant subspace of $V$ is, up to isomorphism, a sum of one or more of the $V_i$.
Conversely, if $V$ is of this form, then $M$ is completely reducible.

In a basis of
$V$ composed of bases of the subspaces $V_{i}$, the matrix of an
element $m$ in $M$ has a diagonal form by blocks,
\begin{displaymath}
  m=\left[\begin{array}{c|c|c|c}
      m_1&\multicolumn{3}{|r}{0}\\ 
      \cline{1-2}
      &m_2&\multicolumn{2}{|}{}\\
      \cline{2-3}
      \multicolumn{2}{c|}{}&\ddots& \\
      \cline{3-4}
      \multicolumn{3}{l|}{0}&m_k
    \end{array}
  \right].
\end{displaymath}
The subspaces $V_i$ are called the \emph{irreducible components}
of $V$ under $M$. The restrictions of $M$ to the subspaces $V_i$
are called the \emph{irreducible constituents} of $M$.

Let $M$ be a monoid and let $V$ be a finite dimensional vector space
over a field $K$. A \emph{linear representation}  of $M$
over $V$ is a morphism $\varphi$ from $M$ into $\End(V)$. A subspace $W$
of $V$ is invariant under $\varphi$ if it is invariant under
$\varphi(M)$.
The representation  is completely
  reducible if the monoid $\varphi(M)$ is completely reducible.
The irreducible components and the irreducible constituents
 of $\varphi$ are those of $\varphi(M)$.

An algebra is said to be \emph{simple} if it has no other two-sided
ideals than $0$ and itself. It is said to be \emph{semisimple} if it
is
a  finite direct product of simple algebras. 

We will use some well-known properties of semisimple algebras 
(see~\cite{CliffordPreston1961}
or \cite{Lang2002} for example).

 First a
quotient
of a  semisimple algebra is  semisimple and a direct product of
semisimple
algebras is semisimple. Note that a subalgebra of a semisimple algebra
is not semisimple in general since any finite dimensional algebra
over a field $K$ is a subalgebra of the algebra $K^{n\times n}$
of $n\times n$-matrices, which is simple.

Next, by Wedderburn's theorem,
every semisimple algebra contains an identity element. 
Finally a nilpotent ideal of a semisimple
algebra is zero.

A representation of an algebra $\Agoth$ over a vector
space $V$ is a morphism $\varphi$ from $\Agoth$  into the
algebra
$\End(V)$. It is \emph{faithful} if $\varphi$ is injective.
The representation is reducible or completely reducible if $\varphi(\Agoth)$
is reducible or completely reducible respectively.

Let $\Agoth$ be an algebra over a field $K$. Let $\varphi$ be a
representation
of $\Agoth$ over a vector space $V$. Then $(v,x)\mapsto v\varphi(x)$
defines a structure of $\Agoth$-module on $V$. Conversely, if $V$ is
an $\Agoth$-module, the map $\varphi:\Agoth\rightarrow \End(V)$ defined
by $v\varphi(x)=vx$ is a linear representation of $\Agoth$ over $V$.
Thus, given $\Agoth$ and $V$, one 
 speaks indifferently of a linear representation of $\Agoth$
over $V$ or of a structure of $\Agoth$-module on $V$.

As well known, the properties of an algebra of being simple or
semisimple correspond to the properties of their representations
to be irreducible or completely reducible respectively. More precisely,
if an algebra has a faithful irreducible representation, then
it is simple. If it has a faithful completely reducible
representation, then it is semisimple. Conversely, any representation
of a simple algebra is irreducible and every representation of
a semisimple algebra is completely reducible (see~\cite{CliffordPreston1961} for example).

We recall that, by Maschke's theorem a linear representation of a
finite
group $G$ over a field $K$ of characteristic not dividing
the order of $G$ is completely reducible.

\subsection{A sufficient condition for complete reducibility}

Let $\Agoth$ be an algebra and $e\ne 0$ be an idempotent of
$\Agoth$. Then $e\Agoth e$
is an algebra. For an $\Agoth$-module
$V$,
the space $Ve$ is an $e\Agoth e$-module called the \emph{condensed module}
of $V$ and $e$ the \emph{condensation idempotent}.
The map from $V$ to $Ve$ is called in \cite{Green2007} 
the \emph{Schur functor} and the following statement is (6.2b).

\begin{proposition}\label{lemmaGreen}
If $V$ is a finite dimensional irreducible $\Agoth$-module
such that $Ve\ne 0$, then $Ve$ is an irreducible $e\Agoth e$-module.
\end{proposition}

\begin{proof}
Let $W$ be a nonzero $eSe$-submodule of $Ve$. Then $W=We$ since $e$
is idempotent. Moreover $W\Agoth$ is a nonzero $\Agoth$-submodule of $V$,
which implies $W\Agoth=V$ since $V$ is irreducible (here and in the sequel,
we denote by $W\Agoth$ the subspace generated by the $wr$ for $w\in W$
and $r\in \Agoth$). Thus 
\begin{displaymath}
Ve=(W\Agoth)e=(We\Agoth)e=W(e\Agoth e)\subset W,
\end{displaymath}
which proves that $W=Ve$.
\end{proof}
The following statement does not seem to have been explicitly
stated before.

\begin{theorem}\label{theoremCompleteReducibility}
Let $\Agoth$ be a  finite dimensional algebra and let $e\in \Agoth$
be an idempotent. Let $V$ be a finite dimensional $\Agoth$-module.
Then the following are equivalent.
\begin{enumerate}
\item[\rm(i)] $V=\bigoplus_{i=1}^m V_i$ with the $V_i$ irreducible
  $\Agoth$-modules
and $V_ie\ne 0$ for $1\le i\le m$.
\item[\rm(ii)] $Ve$ is completely reducible over $e\Agoth e$, $V=Ve\Agoth$
and $\{v\in V\mid v\Agoth e=0\}=0$.

\end{enumerate}
Moreover, if (i) holds, then $Ve=\bigoplus_{i=1}^m V_ie$ with the
$V_ie$ irreducible $e\Agoth e$-modules.
\end{theorem}
\begin{proof}
Assume first that (i) holds. 
Then $Ve=\bigoplus_{i=1}^m V_ie$. Each of
the $V_ie$ is an irreducible  $e\Agoth e$-submodule of $V_i$
by Proposition~\ref{lemmaGreen}. So $Ve$ is completely reducible. Also
$Ve\Agoth=\sum_{i=1}^mV_ie\Agoth=V$ because $V_ie\Agoth=V_i$ as $V_i$ is irreducible.
Finally, let $W=\{v\in V\mid v\Agoth e=0\}$. Since $W$ is invariant, it is
isomorphic
to a direct sum of some of the $V_i$. But $We=0$. If $W\ne 0$, this contradicts
$V_ie\ne 0$ for all $i$. Thus (ii) holds and also the final assertion
of the theorem.

Next assume that (ii) holds. Let $V'$ be an invariant subspace of $V$.
Then $W'=V'e$ is a  subspace of $Ve$ invariant by $e\Agoth e$. Thus it has a
complement $W''$ in $Ve$ invariant by $e\Agoth e$. 
Then $V''=W''\Agoth$ is an invariant subspace which is a complement of $V'$. Indeed,
since $V=Ve\Agoth$ and $Ve=W'+W''$, we have $V=Ve\Agoth=W'\Agoth+W''\Agoth$. 
Since $W'\Agoth=V'e\Agoth$, the first term
is included in $V'$.  Thus $V=V'+V''$.
 Next if $v\in V'\cap V''$,
then  $v\Agoth e\subset W'\cap W''$ and thus $v\Agoth e=0$ which implies $v=0$
by the last assertion of (ii). This shows that $V'\cap V''=0$
and thus that $V=V'\oplus V''$.

Thus $V$ is completely reducible. Set $V=\bigoplus_{i=1}^mV_i$
with the $V_i$ irreducible $\Agoth$-modules. If $V_ie=0$, then any $v\in
V_i$
is in the set $\{v\in V\mid v\Agoth e=0\}$. Thus $V_i=0$, a contradiction.
\end{proof}
We shall use the following consequence of
Theorem~\ref{theoremCompleteReducibility} 
in Section~\ref{sectionBirecurrent} (see also the examples given there).
\begin{corollary}\label{corollaryCR}
Let $\Agoth$ be a finite dimensional algebra and $e\in\Agoth$
be an idempotent. Let $V$ be finite dimensional \textcolor{red}{$\Agoth$}-module
such that $Ve$ is completely reducible over $e\Agoth e$, $V=Ve\Agoth$
and $\{v\in V\mid v\Agoth e=0\}=0$. Then $V$ is completely
reducible.
\end{corollary}
\section{Automata and syntactic representations}\label{SectionSyntactic}
In this section, we introduce the basic terminology concerning automata and
syntactic
representations. For a more detailed exposition, see~\cite{Eilenberg1974} or \cite{BerstelReutenauer2011}.
\subsection{Words and formal series}
Let $A$ be a finite set called an alphabet.
We denote by $A^*$ the set of words on $A$ and by $A^+$ the set of
nonempty words.

For a word $w=a_1a_2\cdots a_n$ with $a_i\in A$, we denote by
$\tilde{w}$
the \emph{reversal} of $w$. By definition $\tilde{w}=a_n\cdots
a_2a_1$. By convention $\tilde{1}=1$.
For a set $X$ of words, the reversal of $X$ is the set 
$\tilde{X}=\{\tilde{x}\mid x\in X\}$.

For two words $x,y\in A^*$, we define $x^{-1}y=z$  if $y=xz$ and 
$x^{-1}y=\emptyset$ otherwise. Symmetrically $xy^{-1}=z$ if $x=zy$ and 
$xy^{-1}=\emptyset$ otherwise. The notation is extended to sets by linearity.
Thus for example $x^{-1}Y=\{z\in A^*\mid xz\in Y\}$.

A word $v$ is a factor of a word $x$ if $x=uvw$ for some words $u,w$.
For a set $X$ of words, we denote by $F(X)$ the set of factors of
the words of $X$.

Let $K$ be a field.
A \emph{formal series} $S$ on the alphabet $A$ with coefficients in $K$ is
a map $S:A^*\rightarrow K$. For $w\in A^*$, we denote by $(S,w)$ the
value of $S$ on $w$. The value $(S,1)$ is called the \emph{constant
  term}
of $S$.

The sum of $S$ and $T$ is the series defined
by $(S+T,w)=(S,w)+(T,w)$. Likewise, for $\alpha\in K$, the series
$\alpha S$ is defined by $(\alpha S,w)=\alpha(S,w)$. In this way
the set of formal series becomes a vector space.

We denote by $K\langle A\rangle$ the free algebra on $A$. Its
elements, called \emph{polynomials}, are formal series such that all but a finite number of
the coefficients are $0$. When the alphabet has one letter $a$, we use
the traditional notation $K[a]$ rather than $K\langle a\rangle$.

For a set $X\subset A^*$, we denote by $\u(X)$ the characteristic
series of $X$, which is defined by $(\u(X),x)=1$ if $x\in X$
and $0$ otherwise.

 Let $n\ge 1$ be an integer. Let
$\lambda$
be a row $n$-vector, let $\mu$ be a morphism from $A^*$ into the monoid
of $n\times n$-matrices and let $\gamma$ be a column $n$-vector,
all with coefficients in $K$. The  triple $(\lambda,\mu,\gamma)$
is said to be a \emph{linear representation} of a series $S$ if for
any word $w$,
\begin{displaymath}
(S,w)=\lambda \mu(w)\gamma.
\end{displaymath}
We also say that  $(\lambda,\mu,\gamma)$ \emph{recognizes} $S$. The
vector $\lambda$ is called the \emph{initial vector} and $\gamma$
the \emph{terminal vector}.
The series $S$ is said to be \emph{rational} if it has a linear
representation.

 We say that a morphism
 $\psi$
 from
the free algebra $K\langle A\rangle$ into  an algebra $\Agoth$
\emph{recognizes} a series $S$ if there is  a linear map
 $\pi:\Agoth\rightarrow K$ such that
$(S,w)=\pi(\psi(w))$ for all $w\in A^*$.

Let $S$ be a rational series and let $(\lambda,\mu,\gamma)$ be
a linear representation of $S$. Then $\mu$ extends to a morphism
from $K\langle A\rangle$ into the algebra $K^{n\times n}$
of $n\times n$-matrices
with coefficients in $K$. This morphism recognizes $S$ since
the linear map  $\pi:K^{n\times n}\rightarrow K$ defined by
$\pi(m)=\lambda m\gamma$ satisfies
$(S,w)=\pi(\mu(w))$ for any $w\in A^*$. 

Conversely, one can recover a linear representation of 
a rational series $S$ from
a morphism $\psi$ into a finite dimensional algebra $\Agoth$
recognizing $S$. Indeed, let us choose a basis of 
$\Agoth$. Then the map $x\rightarrow x\psi(w)$ from $\Agoth$ into
itself
is linear.
It is represented by an $n\times n$-matrix
$\mu(w)$. Let $\lambda$ be the $n$-vector representing $\psi(1)$
in this basis. Let $\gamma$ be a column $n$-vector
$\gamma$ such that $\pi(x)=x\gamma^t$ for any $x$ in $\Agoth$.
Then $\lambda\mu (w)\gamma=(S,w)$ for all $w\in A^*$,
showing that $(\lambda,\mu,\gamma)$ is a linear representation
of $S$.

This shows the following useful equivalent definition of rational
series.

\begin{proposition}\label{propEquivDefRationalSeries}
A series is rational if and only if it can
be recognized by a morphism into a finite dimensional algebra.
\end{proposition}

\begin{example}
Let $X=(a^2)^*$ and let $S$ be the series $S=\u(X)$. Then $S$ is
recognized
by the linear representation $(\lambda,\mu,\gamma)$ with
$\lambda=[1\ 0]$, $\gamma=[1\ 0]^t$ and $\mu$ defined by
$\mu(a)=\begin{bmatrix}0&1\\1&0\end{bmatrix}$. It is also
recognized by the morphism $\mu$ with the linear map
from $K^{2\times 2}$ in $K$ defined by $\pi(x)=x_{1,1}$.
\end{example}

\subsection{Automata}
An automaton $\A=(Q,I,T)$ on the alphabet $A$
is composed with a finite set $Q$ of states, a set $I\subset Q$
of initial states, a set
$T\subset Q$ St terminal states and a set $E\subset Q\times A\times Q$
of edges. If $(p,a,q)$ is an edge we say that it starts at $p$, it
ends at $q$ and that $a$ is its label. 
Two edges $(p,a,q)$ and
$(p',a',q')$
are consecutive if $q=p'$.
A \emph{path} from $p$ to $q$ in the automaton is a sequence 
$c:p\edge{a_1}q_1\rightarrow\cdots\rightarrow q_{n-1}\edge{a_n}q$
of consecutive edges. The word
$w=a_1\cdots a_n$ is its label. We denote such a path $c:p\edge{w}q$.
A word $w$ is recognized by the automaton $\A$ if there is a path
labeled $w$ from a state in $I$ to a state in $T$. Two
automata are called \emph{equivalent} if they recognize the same set
of words.

A set of words $X$ is \emph{rational} if it is the set of words
recognized by an automaton.

An automaton $\A=(Q,I,T)$ is \emph{deterministic} if $\Card(I)\le 1$ and for each
$p\in Q$ and each $a\in A$ there is at most one edge starting at
$p$ and labeled $a$. In this case, there is a partial map from $Q\times A$ to
$Q$
denoted $(q,a)\mapsto q\cdot a$. The maps  $(q,a)\mapsto q\cdot a$
are called the \emph{transitions} of the automaton.
 This map is extended to a partial
map from $Q\times A^*$ to $Q$ also denoted $(q,w)\mapsto q\cdot w$
for $q\in Q$ and $w\in A^*$.

A deterministic automaton with a unique initial state $i$ will be
denoted
$\A=(Q,i,T)$. This notation implies in particular that $Q$ is not
empty.
Unless otherwise stated, all automata considered in this paper are
deterministic.

Let  $\A=(Q,i,T)$ be a deterministic automaton. A state $q\in Q$ is
\emph{accessible} if there is a word $u$ such that $i\cdot u=q$ and
\emph{coaccessible} if there is a word $v$ such that $q\cdot v\in T$.
The automaton is \emph{trim} if every state is accessible and
coacessible.
For any automaton, $\A=(Q,i,T)$ the automaton obtained by suppressing
all states which are not accessible is called the
\emph{accessible part} of $\A$. The automaton obtained
by suppressing all states which are not accessible 
and coacessible is the \emph{trim
part} of $\A$. Both automata are equivalent to $\A$.

Any automaton can be converted into a deterministic equivalent one.
 Indeed, let $\A=(Q,I,T)$ be an automaton with a set $E$ of edges. Let $\B$ be the
automaton
having as states the nonempty subsets of $Q$. Its transitions are defined, for
$U\subset Q$ and $a\in A$, by
$U\cdot a=\{q\in Q\mid (u,a,q)\in E \text{ for some $u\in U$}\}$ if
this set
is nonempty.
Using the set $I$ as initial state and the family ${\cal T}=\{U\subset
Q\mid U\cap T\ne\emptyset\}$ as set of terminal states, one obtains a
deterministic automaton equivalent to $\A$. The
above
construction is called the \emph{subset construction}. The automaton
obtained by taking the accessible part of the result is said to be
obtained by the \emph{accessible subset construction}.

\begin{example}
Let $\A$ be the automaton represented in Figure~\ref{figBidelay1} on the left.
The initial state is $1$ which is the unique terminal state. An
initial
state is indicated by an incoming edge and a terminal state by an
outgoing one. The automaton is not deterministic because there are two
edges
labeled $a$ going out of state $1$. The result of the accessible subset
construction
is indicated in Figure~\ref{figBidelay1} on the right.
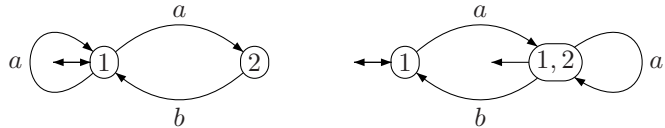
\begin{figure}[hbt]
\centering\gasset{Nadjust=wh}
\begin{picture}(60,10)(0,-5)
\put(0,0){
\begin{picture}(20,10)
\node[Nmarks=if,fangle=180](1)(0,0){$1$}\node(2)(20,0){$2$}

\drawloop[loopangle=180](1){$a$}\drawedge[curvedepth=5](1,2){$a$}\drawedge[curvedepth=5](2,1){$b$}
\end{picture}
}
\put(40,0){
\begin{picture}(20,10)
\node[Nmarks=if,fangle=180](1)(0,0){$1$}\node[Nmarks=f,fangle=180](12)(20,0){$1,2$}

\drawloop[loopangle=0](12){$a$}\drawedge[curvedepth=5](1,12){$a$}
\drawedge[curvedepth=5](12,1){$b$}
\end{picture}
}
\end{picture}
\caption{A nondeterministic automaton and the result of the accessible
  set construction.}\label{figBidelay1}
\end{figure}
\end{example}
We will occasionally consider  a  more general notion of automaton,
called a \emph{weighted automaton}. Let $K$ be a field.
A weighted automaton $\A=(Q,I,T)$ with weights in $K$ is given by 
two maps $I,T:Q\rightarrow K$ and a map $E:Q\times a\times
Q\rightarrow K$.
If $E(p,a,q)=k\ne 0$, we say that $(p,a,q)$ is an edge with label
$a$ and weight $E(p,a,q)$ and we write $p\edge{ka}q$. 

For a given automaton $\A$, any word $w\in A^*$ defines a partial
map $\varphi_\A(w):q\mapsto q\cdot w$ from $Q$ into itself.
The monoid $\varphi_\A(A^*)$ is the
\emph{transition monoid} of the automaton.

The \emph{minimal automaton} of a set $X\subset A^*$ is the automaton
with set of states the nonempty sets $u^{-1}X$ for $u\in A^*$ and
the transitions $(u^{-1}X,a)\mapsto (ua)^{-1}X$. The state $1^{-1}X=X$
is the initial state and the set of final states is the set of
$u^{-1}X$
such that $u\in X$ or, equivalently, $1\in u^{-1}X$. The minimal
automaton of $X$ is trim and recognizes $X$.

The minimal automaton of the empty set has an empty set of
states. In all other cases, the minimal automaton has a nonempty
set of states and a unique initial state.

Let $X\subset A^*$. The set of \emph{contexts} of a word $w$ is the set
$C(w)=\{(u,v)\in A^*\mid uwv\in X$\}. The \emph{syntactic congruence}
is the equivalence  defined by $w\equiv w'$ if
$C(w)=C(w')$ and the \emph{syntactic morphism} is the corresponding
morphism.
 The \emph{syntactic monoid}
of $X$ is the quotient of $A^*$ by the syntactic congruence. 
One sometimes needs to use rather the
 \emph{syntactic semigroup} of a set $X\subset A^+$ which
is the quotient of $A^+$ by the syntactic congruence.

The syntactic monoid of a set $X$ is isomorphic with the transition
monoid of the minimal automaton of $X$.

Let $\A=(Q,i,T)$ be a deterministic automaton and let $K$ be a field. 
The \emph{linear
  representation
associated} with $\A$ is the morphism from $A^*$ into the monoid of
$Q\times Q$-matrices with coefficients in $K$ defined for $a\in A$ by
\begin{displaymath}
\mu(a)_{p,q}=\begin{cases}1&\text{if $p\cdot a=q$}\\0&\text{otherwise}
\end{cases}
\end{displaymath}
If $\A$ is a weighted automaton, we set
$\mu(a)_{p,q}=\sum_{p\edge{ka}q}k$ for each $a\in A$ and $p,q\in Q$.
The \emph{trace} of a word $w$ with respect to the automaton $\A$
is the trace of the matrix $\mu(w)$.

The characteristic series of a rational set of words is a rational
series. Indeed, let $\A=(Q,i,T)$ be an automaton recognizing $X$.
Let $\lambda$ be the characteristic row $Q$-vector of $i$, let
$\mu$ be the linear representation associated with $\A$ and
let $\gamma$ be the characteristic column $Q$-vector of the set $T$.
Then $(\lambda,\mu,\gamma)$ is a linear representation of $\u(X)$.

A deterministic automaton $\A=(Q,i,T)$ is \emph{strongly connected}
if for any $p,q\in Q$ there is a word $w$ such that $p\cdot w=q$.

The following statement is well known (see for example~\cite{Lallement1979}
Proposition 8.2.5
or~\cite{BerstelPerrinReutenauer2009} Proposition 1.12.9). For the
reader's
convenience, we include the proof.
\begin{proposition}\label{propositionMinimalIdeal}
Let $\A=(Q,i,T)$ be a deterministic strongly connected automaton. The transition
monoid $M$ of $\A$ has a unique $0$-minimal or minimal two-sided
ideal $D$ according
to the case where $M$ has a zero or not, formed of the elements
of nonzero minimal rank. It is a regular $\GD$-class. For any
$m\in D$, either $m^2=0$ or the $\GH$-class of $m$ is a group.
\end{proposition}
\begin{proof}
We prove the statement when $M$ has a zero which is the empty map. The other case is similar.
Let $r$ be the minimal nonzero rank of the elements of $M$
and let $D$ be the set of elements of rank $r$.

The set $D\cup 0$  is clearly a two-sided ideal.
We first remark that for any $m,n\in D$, there is an $u\in M$ such that $mun\ne 0$.
Indeed, let $p,q,r,s\in Q$ be such that $pm=q$ and $rn=s$.
Since $\A$  is strongly connected there is an element $u\in M$ such
that $qu=r$. Then $pmun=s$ and thus $mun\ne 0$. 

Let us first show that for each $m\in D$, the right ideal $mM$ is
$0$-minimal. Let $u\in M$ be such that $mu\ne 0$. By the above remark, 
there exists $v\in M$ such that $muvm\ne 0$. Let $I=Qm$ be the image
of $m$ and let $z=uvm$. Since
$z\in Mm$, we have $Iz\subset I$. And since $z\ne 0$,
$Iz\subset I$ implies $Iz=I$ by minimality. Thus there is an integer
$k\ge 1$ such that $z^k$ is the identity on $I$. Set $e=z^k$ and
$w=vmz^{k-1}$. Since $e$ is the identity on $Qm$, one has $me=m$. Thus
$m=muw$ showing that $mu\in R(m)$. 

Let $m,n\in D$.
By the above remark there is $u\in M$ such that
$mun\ne 0$. Since $mM$ is a $0$-minimal right ideal, there
exists $v\in M$ such that $munv=m$. Thus $mun\in R(m)$.
The proof that $mun\in L(n)$ is symmetrical. This shows that $D$
is a $\GD$-class. Since two elements in the same $\GD$-class generate
the same two-sided ideal, $D$ is a $0$-minimal two-sided ideal.

Set $n'=un$. Then $mn'\in R(m)\cap L(n')$ implies by Clifford-Miller
Lemma
that $R(n')\cap L(m)$ contains an idempotent. 
Thus $D$ is a regular $\GD$-class. 

For any $m\in D$ either $m^2\in H(m)$ and $H(m)$ is a group by 
Clifford-Miller
Lemma, or $m^2\notin H(m)$. In this case, since $mM$ is a $0$-minimal
right ideal, we cannot have $m^2\ne 0$. Thus $m=0$.
\end{proof}
Note that if the transition monoid $M$ of a strongly connected
automaton contains a zero, then it is the empty map.
\subsection{Syntactic representations}

Let $S$ be a formal series.
For $u\in A^*$, we denote by
$S\cdot u$ the series defined by $(S\cdot u,v)=(S,uv)$. The following
formulas hold
\begin{displaymath}
S\cdot 1=S,\quad (S\cdot u)\cdot v=S\cdot uv.
\end{displaymath}
The \emph{syntactic space} of $S$, denoted
$V_S$, is the vector space generated by the series $S\cdot u$ for $u\in
A^*$. The \emph{syntactic representation} of $S$ is the morphism
$\psi_S: K\langle A\rangle\rightarrow \End(V_S)$  defined for $x\in V_S$ and $u\in A^*$ by 
\begin{displaymath}
x\psi_S(u)=x\cdot u
\end{displaymath}

The \emph{syntactic algebra} of $S$, denoted $\Agoth_S$,
 is the image of the free algebra
$K\langle A\rangle$ by the syntactic representation. 
The syntactic algebra of $S$
can also be defined directly as follows. Denote by $p\mapsto (S,p)$
the extension of $S$ to the free algebra on $A$. Then $\Agoth_S$ is
the quotient of the free algebra by
the equivalence
\begin{equation}
p\equiv 0\,\Leftrightarrow \,(S,upv)= 0 \text{ for all } u,v\in
A^*.
\label{eqAlgSynt}
\end{equation}
The morphism $\psi_S$ recognizes $S$ since the map
$\pi$ from $\Agoth_S$ into $K$ defined by $\pi(\psi_S(w))=(S,w)$
is well-defined and linear.

The syntactic algebra of a
series $S$
satisfies the following universal property 
(see~\cite{BerstelReutenauer2011}, Exercise 2.1.4). 
\begin{proposition}\label{propUniversal}
If  $\psi:K\langle A\rangle\rightarrow \Agoth$ is a surjective
morphism recognizing $S$,
there exists
a morphism $\rho$ from $\Agoth$  onto the syntactic algebra of $S$
such that $\psi_S=\rho\circ \psi$. 
\end{proposition}
\begin{proof}
Let $\pi:\Agoth\rightarrow K$ be the linear map such that
$\pi(\psi(w))=(S,w)$ for any $w\in A^*$. We have to prove that
if $p\in K\langle A\rangle$ is such that $\psi(p)=0$, then
$\psi_S(p)=0$. But if $\psi(p)=0$, then for any $u,v\in A^*$,
$(S,upv)=\pi(\psi(upv))=\pi(\psi(u)\psi(p)\psi(v)))=0$ and thus
$\psi_S(p)=0$.
\end{proof}
 Thus, in view of Proposition~\ref{propEquivDefRationalSeries},
 a series is rational if and only if its syntactic algebra
is finite dimensional. 

A linear representation $(\lambda,\mu,\gamma)$ 
of a series $S$ is said
to be \emph{minimal} if the dimension of the matrices
$\mu(w)$ is equal to the dimension of $V_S$.
In this case, for any word $w$, $\mu(w)$ is the matrix representing
$\psi_S(w)$ in some basis.

Note that in this case $K^n$ is generated by the vectors
$\lambda\mu(w)$ for $w\in A^*$.
Symmetrically the space $W$ of column $n$-vectors is generated by the
$\mu(w)\gamma$ for $w\in A^*$.

\begin{example}\label{exampleSyntactic}
Let $A=\{a\}$ and $S=\u(a^+)$. Then $\{S,S\cdot a\}$ is a basis of $V_S$ and $S$ is
recognized
by the linear representation $(\lambda,\mu,\gamma)$ with
$\lambda=[1\ 0]$,
$\gamma=[0\ 1]^t$ and
\begin{displaymath}
\mu(a)=\begin{bmatrix}0&1\\0&1\end{bmatrix}.
\end{displaymath}
This representation is minimal.
\end{example}

The following is Proposition 14.7.1 in~\cite{BerstelPerrinReutenauer2009}.
\begin{proposition}\label{st14.7.1}
  Let $X$ be a subset of $A^*$ and let $S = \u(X)$. Let $\varphi$
  be the canonical morphism from $A^*$ onto the syntactic monoid $M$ of
  $X$. Then for all $u,v \in A^*$,
  \begin{displaymath}
    \varphi(u) = \varphi(v) \Leftrightarrow \psi_{S}(u) =
    \psi_{S}(v)\,. 
  \end{displaymath}
  In particular the monoid $\psi_{S}(A^*)$ is isomorphic to $M$.
\end{proposition}


\section{Completely reducible sets}\label{sectionCompletelyReducible}
In this section we define the family of completely reducible sets.
In the first part, we prove some closure properties of this family
(Theorem~\ref{theoremVariety}). In the second part,
we prove some necessary and some sufficient conditions for membership
in the family and a characterization in the case of a one-letter alphabet.

\subsection{Completely reducible series}
A series is \emph{completely reducible} if its syntactic
  representation is completely reducible.
As we have seen in Section
\ref{sectionMonoidsMatrices}, this is equivalent to the semisimplicity
of its syntactic algebra. Moreover, by Proposition~\ref{propUniversal},
if a series $S$ is recognized by a morphism onto a semisimple
algebra, then $S$ is completely reducible.

The following result was suggested to me by Christophe Reutenauer
(personal communication).
\begin{proposition}\label{propositionChristophe1}
Any linear combination of completely reducible series is 
completely reducible.
\end{proposition}
We use the following property.

\begin{lemma}\label{lemmaSubDirect}
Let $\varphi_1,\varphi_2$ be two morphisms from $A^*$ into
$\End(V_1)$ and $\End(V_2)$ respectively. Set $V=V_1\times V_2$.
 If $\varphi_1,\varphi_2$
are completely reducible, the morphism $\varphi$ from $A^*$
into $\End(V)$ defined by 
$\varphi(w)(v_1,v_2)=(\varphi_1(w)(v_1),\varphi_2(w)(v_2))$
is completely reducible.
\end{lemma}

\begin{proof}
Since $V_1$ and $V_2$ are direct sums of irreducible components $W_i$,
the same holds for $V$ and thus $\varphi$ is completely reducible.
\end{proof}
We now give the proof of Proposition~\ref{propositionChristophe1}.
\begin{proof}
If $S$ is completely reducible, then $\alpha S$ is clearly completely reducible
for any $\alpha\in K$. Next let $S_1,S_2$ be completely reducible
series. For
$i=1,2$, let $\Agoth_i=\Agoth_{S_i}$, $\psi_i=\psi_{S_i}$ and let
$\pi_i:\Agoth_i\rightarrow K$ be the linear map defined by
$\pi_i(\psi_i(w))=(S_i,w)$.
Consider the morphism $\psi:A^*\rightarrow \Agoth_1\times \Agoth_2$
defined by $\psi(w)=(\psi_1(w),\psi_2(w))$. It recognizes $S+T$ since
the  map $\pi:\Agoth_1\times \Agoth_2\rightarrow K$ defined by 
$\pi(x)=\pi_1(x)+\pi_2(x)$ is linear and such that 
$\pi(\psi(w))=(S_1,w)+(S_2,w)$.
By Lemma~\ref{lemmaSubDirect},
$\psi$ is completely reducible. Thus $\psi(K\langle A\rangle)$
is semisimple which implies that $S_1+S_2$ is completely reducible.
\end{proof}
\subsection{The family of completely reducible sets}
The syntactic representation (resp. algebra) of a set $X\subset A^*$ is the
syntactic representation (resp. algebra) of its characteristic
series. 

A rational set
is \emph{completely reducible} if its characteristic series
is completely reducible. 
\begin{example}\label{exampleCompletlyReducible1}
The sets $a^*$, $a^+$ and $1$ are completely reducible. Indeed,
the syntactic algebras of $a^*$ and $1$ have dimension $1$.
Concerning  $a^+$,
the linear representation of Example~\ref{exampleSyntactic}
takes in the basis $S-S\cdot a,S\cdot a$ the form $(\lambda',\mu',\gamma')$ with $\lambda'=[1\ 1]$,
$\gamma'=[-1\ 1]^t$ and
\begin{displaymath}
\mu'(a)=\begin{bmatrix}0&0\\0&1\end{bmatrix}.
\end{displaymath}
\end{example}
\begin{example}\label{example(ab)^*}
The sets $X=(ab)^*$ and  $Y=(ab)^*a$ are completely reducible. Indeed, 
$\u(X)$ is recognized by the linear representation
$(\lambda,\mu,\gamma)$
with $\lambda=[1\ 0]$ $\gamma=[1\ 0]^t$ and
\begin{displaymath}
\mu(a)=\begin{bmatrix}0&1\\0&0\end{bmatrix},\quad
\mu(b)=\begin{bmatrix}0&0\\1&0\end{bmatrix}
\end{displaymath}
Since there are no nontrivial invariant subspaces, the representation
$\mu$ is completely reducible.
The series $\u(Y)$
is recognized by $(\lambda,\mu,\gamma')$ with $\gamma'=[0\ 1]^t$.

\end{example}
\begin{example}\label{examplea}
The set $X=a$ is not completely reducible. Indeed,  $\u(X)$
is recognized by the linear
representation 
 $(\lambda,\mu,\gamma)$ with $\lambda=[1\ 0]$,
$\gamma=[0\ 1]^t$ and
\begin{displaymath}
\mu(a)=\begin{bmatrix}0&1\\0&0\end{bmatrix}.
\end{displaymath}
This representation is minimal.
The subspace generated by $[0\ 1]$ is the only nontrivial 
invariant subspace. Thus $X$ is not completely reducible.
\end{example}

\begin{theorem}\label{theoremVariety}
The family of completely reducible sets is  closed by residual,
complement
and reversal.
\end{theorem}

\begin{proposition}\label{propReversal}
The family of completely reducible sets is closed by reversal.
\end{proposition}
\begin{proof}
Let $X$ be a completely reducible set. Let $(\lambda,\mu,\gamma)$ be
a linear representation of $\u(X)$. Let $\nu$ be the
morphism from $A^*$ into $K^{n\times n}$ defined by
$\nu(w)=\mu(\tilde{w})^t$.
Then $(\gamma^t,\nu,\lambda^t)$ is a linear representation
of $\u(\tilde{X})$. Indeed
\begin{displaymath}
(\u(\tilde{X}),w)=(\u(X),\tilde{w})=\lambda\mu(\tilde{w})\gamma
=(\lambda\mu(\tilde{w})\gamma)^t=\gamma^t\nu(w)\lambda^t.
\end{displaymath}
Moreover  $(\lambda,\mu,\gamma)$ is minimal if and only if
$(\gamma^t,\nu,\lambda^t)$ is minimal.
\end{proof}
We now give the proof of Theorem~\ref{theoremVariety}.
\begin{proof}
Let $\V$ be the family of completely reducible sets.

For $X\in \V$ and $w\in A^*$, let $Y=w^{-1}X$. Set $S=\u(X)$ and
$T=\u(Y)$. Then $T=S\cdot w$.
The syntactic space $V_T$ is the space generated by the $T\cdot
u=S\cdot wu$. Since $V_T$ is an invariant subspace of $V_S$, the
invariant
subspaces of $V_T$ are invariant subspaces of $V_S$. Thus the
syntactic
representation of $T$ is also completely reducible. Therefore $Y\in \V$.
The proof that $Xw^{-1}\in A^*\V$ is similar, using
Proposition~\ref{propReversal}. 

Let $X\in \V$ and set $Y=A^*\setminus X$. Since $A^*$ is completely reducible, 
by Proposition~\ref{propositionChristophe1}, the  series
$\u(Y)=\u(A^*)-\u(X)$ is completely reducible.
Thus $Y\in\V$.
\end{proof}

The family of completely reducible sets  is not closed by
intersection,
as shown by  Example~\ref{exampleInter}. Since it is closed by
complement is not closed by union either.
\begin{example}\label{exampleInter}
Let $X=(ab)^*a$ and $Y=(ac)^*a$. The sets $X$ and $Y$ are completely
reducible by Example~\ref{example(ab)^*}. We have $X\cap Y=a$
which is not completely
reducible by Example~\ref{examplea}. 
\end{example}
The following result shows an additional closure property.
\begin{proposition}\label{propEmptyWord}
For any rational set $X$, the sets $X$ and $X\cap A^+$
are simultaneously completely reducible.
\end{proposition}
\begin{proof} We may assume that $1\in X$. Set $Y=X\cap A^+$.
Since $\u(Y)=\u(X)-1$, this results directly from Proposition~\ref{propositionChristophe1}.
\end{proof}

\subsection{Some properties of completely reducible sets}

 It follows from Proposition~\ref{st14.7.1}
that the syntactic algebra of a set $X$ is a quotient of the algebra $K[M]$
where $M$ is the syntactic monoid of $X$. As a consequence, we have
the
following statement, which gives a sufficient condition for complete reducibility.
\begin{proposition}
If the algebra of the syntactic monoid of a set $X\subset A^*$
is semisimple, then
 $X$ is completely reducible.
\end{proposition}

The following result gives in turn a necessary condition for
complete reducibility.
\begin{proposition}\label{propositionNecessary}
If $X\subset A^*$ is completely reducible, the syntactic
monoid of $X$ has a faithful completely reducible representation.
\end{proposition}
\begin{proof}
This follows directly from Proposition~\ref{st14.7.1}.
\end{proof}

The semigroups having a faithful completely reducible representation
over $\C$ have been characterized by Rhodes~\cite{Rhodes1969}. The
characterization in arbitrary characteristics 
from~\cite{AlmeidaMargolisSteinbergVolkov2009}
is the following. If $\mathbf{V}$ is a class of semigroups, a
monoid morphism
$\varphi:M\rightarrow N$ is called a $\mathbf{V}$-morphism if
$\varphi^{-1}(e)\in\mathbf{V}$ for any idempotent $e\in N$.
A congruence  on $M$ is a $\mathbf{V}$-congruence
if the corresponding quotient morphism is a $\mathbf{V}$-morphism.
Denote by $\mathbf{G}_K$ the class of groups which is reduced to
the trivial group if the characteristic of $K$ is $0$
and to the class of finite $p$-groups if the characteristic
of $K$ is $p\ne 0$. Denote by $\mathbf{LG}_K$ the class of finite semigroups
$S$ such that $eSe\in \mathbf{G}_K$ for any idempotent $e\in S$.

The main result of~\cite{AlmeidaMargolisSteinbergVolkov2009} says
that  the intersection of the congruences
associated
to the irreducible representations of a finite monoid $M$
is the largest $\mathbf{LG}_K$-congruence.
This congruence is called the \emph{Rhodes radical} of the monoid $M$.
In particular,
a monoid $M$ has a faithful completely
reducible representation if and only if its Rhodes radical
is trivial.

A set $X\subset a^+$ is \emph{periodic}  if there is a 
nonzero
 integer
$n$ such that for each
$i\ge 1$, $a^i\in X$ if and only if $a^{n+i}\in X$. The least such
integer $n$ is called the period
of $X$.

The following result is a characterization of completely reducible
sets on a one letter alphabet. In the proof, we use the following result.
Let $V$ be a finite dimensional vector space over $K$.
An element $x$ in $\End(V)$ generates a semisimple algebra
if and only if the minimal
polynomial of $x$ has no factors of multiplicity $>1$ over $K$
(see~\cite{Lang2002} Chapter XVII, Exercise 10).
\begin{theorem}\label{theorem1letter}
A rational set $X\subset a^*$ is completely reducible if and only if 
$X\cap A^+$ is periodic and the period of $X$ does not divide the
characteristic of $K$.
\end{theorem}
\begin{proof}
Assume first that $X\cap A^+$ is periodic of period $n$
with $n$ not dividing the characteristic of $K$.
The syntactic semigroup $M$ of $X\cap A^+$ is isomorphic to $\Z/n\Z$. 
The algebra
of $M$ is
semisimple
since the algebra of the group $\Z/n\Z$ is semisimple. 
Thus $X\cap A^+$ is completely reducible. By
Proposition~\ref{propEmptyWord}
it implies that $X$ is completely reducible.

Conversely, let $X$ be a nonempty completely reducible subset of $a^+$.
Set $S=\u(X)$, $V=V_S$, $\psi=\psi_S$ and $\Agoth=\Agoth_S$.
Let $\varphi$ be the canonical morphism from $a^*$ onto the
 the syntactic monoid $M$  of $X$ and let $m=\varphi(a)$. Let $i\ge 0$
and $p\ge 1$ be the index of $M$ in such a way that
\begin{displaymath}
M=\{1,m,m^2,\ldots,m^{i-1},m^i,\ldots,m^{i+n-1}\}
\end{displaymath}
with $m^{i+n}=m^i$. 

 Set $x=\psi(a)$. Since $\Agoth$ is a quotient of $K[M]$,
the minimal polynomial $f(t)$ of $x$ divides $t^i(1-t^n)$. Since
$\Agoth$ is semisimple, the factor $t$ has multiplicity at most $1$
and thus $f(t)$ divides $t(1-t^n)$. This shows that $x=x^{n+1}$.

 This implies that, for any $i\ge 1$, $x^i=x^{i+n}$
and thus that 
\begin{displaymath}
a^i\in X\Leftrightarrow (S,a^i)=1\Leftrightarrow (S\cdot a^i,1)=1\Leftrightarrow
(S\cdot a^{i+n},1)=1\Leftrightarrow a^{i+n}\in X.
\end{displaymath}
Thus $X\cap A^+$ is periodic of period $n$. Finally, if the
characteristic
of $K$ is $p\ne 0$, then $n$ cannot be a multiple of $p$ since
otherwise $1-t^n=(1-t^\frac{n}{p})^p$ and $f(t)$ would have factors of
multiplicity $>1$.
\end{proof}

 Theorem~\ref{theorem1letter} implies that
the completely reducible subsets  of $a^*$ are of the form
$X$ or $X\cup 1$ for $X\subset a^+$ periodic.

Note that, in the proof of Theorem~\ref{theorem1letter}, we could
have used the Rhodes radical mentioned above to prove the necessity
of the condition. Indeed, the Rhodes radical of a finite cyclic monoid
generated by $m$ is trivial if
and only if $m=m^n$ for some $n>1$ which is not a multiple of the 
characteristic of $K$.

We end the section with a necessary condition for complete
reducibility. We say that a set $X$ is \emph{repeating}
if for any $x\in X$ there exist words $u,v$ such that
$xuxv\in X$.
\begin{proposition}
Any completely reducible rational set is repeating.
\end{proposition}
\begin{proof}
Arguing by contradiction, assume that $X$ is not repeating. Let $x\in
X$  be
such that $xA^*xA^*\cap X=\emptyset$. Set $S=\u(X)$ and $V=V_S$. 

Let $V'$
be the subspace of $V$ generated by the series $S\cdot xu$ for $u\in
A^*$. Note that for any element $T$ of $V'$, we have $T\cdot x=0$.
Indeed, if $T=\sum_{i=1}^n \alpha_i S\cdot xu_i$ for some $\alpha_i\in
K$, we have 
$T\cdot x=\sum_{i=1}^n \alpha_i S\cdot xu_ix=0$ since for any word $u$,
we have $S\cdot xux=0$.

We have $V'\ne 0$ because $(S\cdot x,1)=1$ and thus $S\cdot x$
is a nonzero element of $V'$. Since $S\cdot x\ne 0$, 
we have $S\notin V'$ and thus $V'\ne V$. By definition, $V'$ is
invariant.
Assume that $V'$ has an invariant complement $V''$. Then $S=S'+S''$
with $S'\in V'$ and $S''\in V''$. Since $S'\cdot x=0$, we have $S\cdot
x=S''\cdot x$. This implies that $S''\cdot x$ is in $V'$. Since
$V''$ invariant, we have also
$S''\cdot x\in V''$  and thus 
$S''\cdot x=0$. This implies that  $S\cdot x=0$, a contradiction. 
Thus $X$ is not
completely reducible.
\end{proof}
\begin{example}
Let $X=ab^*$. For $w=a$, the
set
$X\cap wA^*wA^*$ is reduced to $a$ and therefore $X$ is not repeating.
This shows that $X$ is not completely reducible.
\end{example}
\section{Birecurrent sets}\label{sectionBirecurrent}
In this section we introduce the class of birecurrent sets
and we prove their complete reducibility. In the first
part we state the main result (Theorem~\ref{theoremBifix}).
In the second one we introduce the notion of accessible
reversal of an automaton used in the proof of
Theorem~\ref{theoremBifix}.
In the last part, we give the proof of Theorem~\ref{theoremBifix}.

\subsection{Main result}\label{sectionMainResult}
A nonempty set $X$ is called \emph{recurrent} if its minimal automaton
is strongly connected. It is said to be \emph{birecurrent}
if $X$ and its reverse $\tilde{X}$ are recurrent.

The submonoid generated by a prefix code is recurrent. Indeed,
let $X$ be prefix code. The submonoid generated by $X$ is
right unitary, which means by definition that for any words $u,v$
if $u,uv\in X^*$, then $v\in X^*$. This implies that for any $x\in
X^*$,
one has $x^{-1}X^*=X^*$. Thus the minimal automaton of $X^*$ is of the
form $\A=(Q,i,i)$ with a set of terminal states reduced to the initial
state.
Since $\A$ is trim, this implies that $\A$ is strongly connected.

Thus the submonoid generated by a bifix code is birecurrent.
The following example shows that other cases occur.
\begin{example}
Let $X=\{a,ba\}$. The set $X$ is a prefix code which is not bifix. The
submonoid
$X^*$ is birecurrent. Indeed, the minimal automata of $X^*$ and
$\tilde{X}^*$
are represented in Figure~\ref{figBidelay}. Both are strongly
connected.
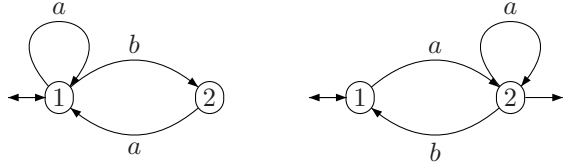
\begin{figure}[hbt]
\centering\gasset{Nadjust=wh}
\begin{picture}(60,20)(0,-5)
\put(0,0){
\begin{picture}(20,10)
\node[Nmarks=if,fangle=180](1)(0,0){$1$}\node(2)(20,0){$2$}

\drawloop(1){$a$}\drawedge[curvedepth=5](1,2){$b$}\drawedge[curvedepth=5](2,1){$a$}
\end{picture}
}
\put(40,0){
\begin{picture}(20,10)
\node[Nmarks=if,fangle=180](1)(0,0){$1$}\node[Nmarks=f](2)(20,0){$2$}

\drawloop(2){$a$}\drawedge[curvedepth=5](1,2){$a$}\drawedge[curvedepth=5](2,1){$b$}
\end{picture}
}
\end{picture}
\caption{The minimal automata of $X^*$ and $\tilde{X}^*$.}\label{figBidelay}
\end{figure}
\end{example}
We will prove the following statement. We assume in this section
that the field $K$ is of characteristic $0$.
\begin{theorem}\label{theoremBifix}
A birecurrent set is completely reducible.
\end{theorem}

Theorem~\ref{theoremBifix} implies the following
result, originally from ~\cite{Reutenauer1981}, where the result is proved with
a partial converse.

\begin{corollary}\label{corollaryBifix}
 The submonoid generated by a rational bifix
code
is completely reducible. 
\end{corollary}

The proof of Theorem~\ref{theoremBifix}  uses
Theorem~\ref{theoremCompleteReducibility}. It is essentially the same as
that
given in~\cite{BerstelPerrinReutenauer2009} for
Corollary~\ref{corollaryBifix}.
\subsection{Accessible reversal of an automaton}
We begin this section with the following definition.

Let $\A=(Q,i,T)$ be a deterministic automaton. The \emph{accessible reversal} of $\A$, denote by
$\tilde{A}$, is the automaton obtained by successively
\begin{enumerate} 
\item[(i)]reversing the edges of $\A$,
\item[(ii)] using the accessible subset construction to build an equivalent
  deterministic automaton  using $T$ as initial state and the
subsets containing $i$ as set of terminal states,
\end{enumerate}

Thus  $\tilde{A}=(\tilde{Q},T,J)$ where $\tilde{Q}$ is
the family of  nonempty sets of the form
\begin{displaymath}
w^{-1}T=\{q\in Q\mid q\cdot w\in T\}
\end{displaymath}
and $J=\{U\in\tilde{Q}\mid i\in U\}$.
This automaton recognizes
$\tilde{X}$. Indeed, $y$ is in $\tilde{X}$ if and only if
$i\cdot \tilde{y}\in T$. And $i\cdot \tilde{y}$ is in $T$ if
and only if $i$ is in $T\cdot y$ (for the transitions of $\tilde{A}$).
Let $M=\varphi_\A(A^*)$ and $\tilde{M}=\varphi_{\tilde{\A}}(A^*)$ be the monoids of transitions of $\A$ and
$\tilde{\A}$
respectively. There is an antiisomorphism $m\mapsto\tilde{m}$
from $M$ onto $\tilde{M}$
such that the diagram below is commutative.
\begin{figure}[hbt]
\centering
\gasset{Nframe=n,Nadjust=wh}
\begin{picture}(20,22)
\node(1)(0,20){$A^*$}\node(2)(20,20){$A^*$}
\node(3)(0,0){$M$}\node(4)(20,0){$\tilde{M}$}

\drawedge(1,2){$\sim$}\drawedge(1,3){$\varphi_\A$}
\drawedge(2,4){$\varphi_{\tilde{\A}}$}
\drawedge(3,4){$\sim$}
\end{picture}
\end{figure}

In particular, for any word $w$, one has $m=\varphi_\A(w)$ if and
only if $\tilde{m}=\varphi_{\tilde{\A}}(\tilde{w})$.

The action of $M$ on the left on $\widetilde{Q}$ defined by $mU=V$
if $V=\{q\in Q\mid qm\in U\}$ is such that 
\begin{equation}
mU=V\Leftrightarrow U\tilde{m}=V. \label{eqTilde}
\end{equation}

The following statement is well known (see~\cite{Eilenberg1974} p. 48).
\begin{proposition}\label{propositionJojo}
If $\A$ is a trim deterministic automaton recognizing $X$, then $\tilde{\A}$
is the minimal automaton of $\tilde{X}$.
\end{proposition}
\begin{proof}
 Since $\A$ is trim, for any word $w$, one has
$w^{-1}T\ne\emptyset$ if and only if $Xw^{-1}\ne\emptyset$.
Moreover, for any
$w,w'\in A^*$, one has
\begin{displaymath}
w^{-1}T=w'^{-1}T\Leftrightarrow Xw^{-1}=Xw'^{-1}
\end{displaymath}
as one may easily verify.
Since the nonempty sets $Xw^{-1}$ are the reversals of the states
of the minimal automaton $\tilde{X}$, the map $w^{-1}T\mapsto \tilde{w}^{-1}\tilde{X}$
is a bijection which identifies $\tilde{\A}$ with the minimal
automaton of $\tilde{X}$.
\end{proof}
Thus, in particular, if $\A$ is the minimal automaton of $X$, then
$\tilde{\A}$ is the minimal automaton of $\tilde{X}$.

\begin{example}\label{exampleBirecurrent}
Let $\A=(Q,i,T)$ with $Q=\{1,2,3,4\}$, $i=1$ and $T=\{1,2\}$
be the strongly connected automaton represented on the left
in Figure~\ref{figureWeakly}.
\begin{figure}[hbt]
\centering\gasset{Nadjust=wh}
\begin{picture}(60,25)
\put(0,0){
\begin{picture}(30,20)
\node[Nmarks=if,fangle=180](1)(0,20){$1$}\node[Nmarks=f](2)(20,20){$2$}
\node(3)(20,0){$3$}\node(4)(0,0){$4$}

\drawedge(1,2){$a$}\drawedge[curvedepth=3](1,3){$b$}
\drawedge(2,3){$a,b$}
\drawedge(3,4){$a$}\drawedge[curvedepth=3](3,1){$b$}
\drawedge(4,1){$a,b$}
\end{picture}
}
\put(40,0){
\begin{picture}(30,20)
\node[Nmarks=if,fangle=180](12)(0,20){$1,2$}\node(23)(20,20){$2,3$}
\node(34)(20,0){$3,4$}\node[Nmarks=f,fangle=180](14)(0,0){$1,4$}

\drawedge[ELside=r](23,12){$a,b$}\drawedge[curvedepth=3](12,34){$b$}
\drawedge[ELside=r](34,23){$a$}
\drawedge[ELside=r](14,34){$a,b$}\drawedge[curvedepth=3](34,12){$b$}
\drawedge[ELside=r](12,14){$a$}
\end{picture}
}
\end{picture}
\caption{The automata $\A$ and $\tilde{\A}$.}\label{figureWeakly}
\end{figure}
The accessible reversal $\tilde{\A}$ of $\A$ is represented in Figure~\ref{figureWeakly}
on the right.
Since $\tilde{\A}$ is strongly connected, $X$ is birecurrent.
Note that $X$ is not a submonoid since $a,abb\in X$ although
$aabb\notin X$.
\end{example}

\subsection{Proof of the main result}\label{sectionBirecurrentSetsCompletelyReducible}

We begin with two preliminary statements.

\begin{proposition}\label{propositionUnitary}
Let $\A=(Q,i,T)$ be the minimal automaton of a set $X$.
Set $S=\u(X)$, $\tilde{S}=\u(\tilde{X})$ and $\varphi=\varphi_\A$. For any word $x\in A^*$, one
has
\begin{enumerate}
\item[\rm (i)] $i\varphi(x)=i$ if and only if $S\cdot x=S$,
\item[\rm(ii)] $\varphi(x)T=T$ if and only if $\tilde{S}\cdot \tilde{x}=\tilde{S}$.
\end{enumerate}
\end{proposition}
\begin{proof}
Assume that $i\cdot x=i$. Then, for any $u\in A^*$,
\begin{displaymath}
(S\cdot x,u)=1\Leftrightarrow
xu\in X\Leftrightarrow i\cdot xu\in
  T\Leftrightarrow i\cdot u\in
  T\Leftrightarrow (S,u)=1.
\end{displaymath}
Thus $S\cdot x=S$. Conversely, if $S\cdot x=S$, then for any $u\in A^*$,
\begin{displaymath}
i\cdot xu\in T\Leftrightarrow (S,xu)=1\Leftrightarrow (S\cdot x,u)=1
\Leftrightarrow (S,u)=1\Leftrightarrow i\cdot u\in T
\end{displaymath}
which implies that $x^{-1}X=X$.
In view of the definition of the minimal automaton, this
shows that $i\cdot x=i$ . Thus proves (i). The proof of (ii) is
the same, using the fact that, by \eqref{eqTilde}, one has
$\varphi(x)T=T$ if and only if
$T\cdot \tilde{x}=T$ in the automaton $\tilde{\A}$.
\end{proof}

\begin{proposition}\label{st8.7.4}  
 Let $\A=(Q,i,T)$ be the minimal automaton of
a  birecurrent set $X$ \textcolor{red}{containing the empty word}. Set $\varphi=\varphi_\A$ and $M=\varphi(A^*)$.
  The monoid $M$
   contains an idempotent $e$ such that
  \begin{itemize}
  \item[\textup{(i)}] $ie=i$ and $eT=T$.
  \item[\textup{(ii)}] The set $eMe$ is the union of a finite group
    $G$ and of the element $0$, provided $0 \in M$.
  \end{itemize} 
\end{proposition}
\begin{proof}
We assume that $M$ contains a zero. The other case is similar.
Since $\A$ is strongly connected, the zero is the empty map $0$.
By Proposition~\ref{propositionMinimalIdeal}, the monoid $M$
has a unique $0$-minimal two-sided ideal $D$ which is a regular $\GD$-class.
Let $w$ be a word such that $\varphi(w)$ belongs to 
 $D$. Since $\A$ is strongly connected there
is a word $u$ such that $i\in Q\cdot wu$. Set $w'=wu$. 
Then $\varphi(w')$ is in $D$. \textcolor{red}{
Moreover,
since $i\in T$ and $i\in Q\cdot w'$, we have 
$T\cdot \tilde{w}'\ne\emptyset$.
}
 Thus, since $\tilde{A}$ is strongly connected, there
is a $v$ such that $T\cdot \tilde{w}'v=T$ and thus
$\varphi(\tilde{v}w')T=T$.  Then $\varphi(\tilde{v}w')$ is in $D$.
Since $\varphi(\tilde{v}w')T=T$, we cannot have
$\varphi(\tilde{v}w')^2=0$
and thus  there is a power $x$ of $\tilde{v}w'$ such that
$e=\varphi(x)$ is an idempotent. Since $i\in Q\cdot w'$, $i$ is in
the image of $e$ and thus we have
$ie=i$ since $e$ is idempotent. We have also $eT=T$. Moreover the non
zero elements of the set $eMe$  form the group of the $\GD$-class $D$. 
\end{proof}

The group $G$ defined above is called the Suschkevitch group of the
monoid $M$. It is the group of the $0$-minimal ideal of $M$.

Note that the formulation of Proposition~\ref{st8.7.4} can be used
to define a birecurrent set by a condition on its syntactic monoid.
Consider indeed $X\subset A^*$, let $M$ be its syntactic monoid
and let $\varphi:A^*\rightarrow M$ be the syntactic morphism.
Then $X$ is recurrent if and only if there is a non-zero idempotent $e\in M$
such that
\begin{enumerate}
\item[(i)] $eMe$ is the union of a group $G$ and of $0$ provided $0\in
  M$.
\item[(ii)] There is a subset $P$ of $G$ such that $X=\{x\in A^*\mid
  e\varphi(x)e\in P\}$.
\end{enumerate}
\begin{example}
Consider again the birecurrent set of
Example~\ref{exampleBirecurrent}.
The minimal ideal of $M$ is represented in Figure~\ref{figureMinimalIdealRec}.
\begin{figure}[hbt]

\begin{displaymath}
\def\rb{\hspace{2pt}\raisebox{0.8ex}{*}}\def\vh{\vphantom{\biggl(}}
    \begin{array}%
    {r|@{}l@{}c|@{}l@{}c|}%
    \multicolumn{1}{r}{}&\multicolumn{2}{c}{1,3}&\multicolumn{2}{c}{2,4}\\
    \cline{2-5}
    1,2/3,4& \vh\rb &b &\vh\rb  &ba \\
    \cline{2-5}
    1,4/2,3&\vh\rb &ab &\vh\rb& aba \\
    \cline{2-5}
    \end{array}
\end{displaymath}
\caption{The $0$-minimal ideal of $M$.}\label{figureMinimalIdealRec}
\end{figure}
The idempotent $e=\varphi_\A(b^2)$ is such that $1e=1$ and $eT=T$. The
set $eMe$ is the group $\Z/2\Z$.
\end{example}
We now give the proof of Theorem~\ref{theoremBifix}.

\begin{proof}
Let $X$ be a birecurrent set and let $S=\u(X)$.
Let $\A=(Q,i,T)$ be the minimal automaton of $X$. 

\textcolor{red}{
We may assume that $i\in T$ or equivalently that $X$
contains the empty word. Indeed, let $t\in T$ and consider the set
$X'$ recognized by the automaton $\A'=(Q,t,T)$. Then $X'$
is a birecurrent set containing the empty word. Moreover,
since $\A$ is strongly connected, $X$ is a residual of $X'$
and thus $X$ is completely reducible if $X'$ is, by Theorem~\ref{theoremVariety}.
}

Set $\varphi=\varphi_\A$ and $\psi=\psi_S$.
  By Proposition \ref{st8.7.4}, there exists a word $x\in A^*$
 such that $\varphi(x)$ is idempotent, $i\varphi(x)=i$ and $\varphi(x)T=T$
and such that $\varphi(xA^*x)$ is the
  union of $0$ (if $0 \in \varphi(A^*)$) and of a finite group.

Set $M=\psi(A^*)$ and $e=\psi(x)$. By Proposition \ref{st14.7.1}, $e$
is an idempotent of $M$ such that $eMe$  is the
  union of $0$ (if $0 \in M$) and of a finite group
(note that $0\in M$ if and only if $\varphi(A^*)$ contains a zero).
Moreover,
by Proposition~\ref{propositionUnitary} and its dual, we have
$(S,u)=(S,ux)=(S,xu)$ for any $u\in A^*$.

Set $V=V_S$. Taking a basis of $V$, we may consider $M$ as
a monoid of $n\times n$-matrices and $V$ as the space of
row $n$-vectors. Let $\lambda$ be the row $n$-vector representing $S$
and let $\gamma$ be the column $n$-vector such that 
$(S,w)=\lambda\psi(w)\gamma$ for all $w\in A^*$.

Set $\Agoth$ be the algebra generated by $M$. Then $V$ is a finite dimensional $\Agoth$-module.
We verify that the conditions of 
Corollary~\ref{corollaryCR} are satisfied by $\Agoth$, $V$ and
$e$. Since $e\Agoth e$ is the algebra generated by $eMe$,
by Maschke's theorem, $Ve$ is completely reducible over $e\Agoth e$.
Next, since $i\varphi(x)=i$,
we have $\lambda e=\lambda$ by Proposition~\ref{propositionUnitary}.
 Since $V$ is generated by the vectors $\lambda m$ for
$m\in M$, it is generated by the set $\lambda eM$. Thus the condition
that $V$ is generated by the set $Ve\textcolor{red}{\Agoth}$ is also satisfied. Finally,
let $W$ be the space of column $n$-vectors. Symmetrically to the fact
that $V$ is generated by the elements of the set $\lambda m$, 
for $m\in M$, the space $W$ is generated by the elements of
the set $m\gamma$ for $m\in M$.
By assertion (ii) of Proposition~\ref{propositionUnitary}, since $\varphi(x)T=T$, we
 have $e\gamma=\gamma$. Thus $W$ is generated by the elements of the
set
$Me\gamma$, which implies that $W$ is the space generated by $\textcolor{red}{\Agoth}eW$. 
Now, one has $v\Agoth e=0$ if and
only if $v\textcolor{red}{\Agoth}eW=0$ and so $\{v\in V\mid v\Agoth e=0\}$ is the orthogonal of
the space generated by $\Agoth eW$. Thus we conclude that 
$\{v\in V\mid v\Agoth e=0\}=0$. This shows that all conditions in (ii)
are satisfied.

By Corollary~\ref{corollaryCR}, the monoid $M$ is
completely reducible and thus the proof is complete.
\end{proof}
Note that for any birecurrent set $X$, by Theorem~\ref{theoremCompleteReducibility}, the
irreducible
components of the syntactic representation of $X$ are in bijection
with the irreducible components of the permutation representation
of the group. We illustrate this in the following example.
\begin{example}
Consider again the birecurrent set of
Example~\ref{exampleBirecurrent}.
The syntactic representation of $X$ is obtained from the linear
representation associated with the automaton $\A$ after taking
the quotient of the space $K^Q$ by the subspace generated by
$1+3-2-4$. Thus, in the basis $1,2,3$, we have
\begin{displaymath}
\psi(a)=\begin{bmatrix}0&1&0\\0&0&1\\1&-1&1\end{bmatrix},\quad
\psi(b)=\begin{bmatrix}0&0&1\\0&0&1\\1&0&0\end{bmatrix}
\end{displaymath}
The subspace generated by the vector $1+3$ is invariant. It has
an invariant complement formed of the vectors with zero sum of
coefficients.
In the basis $1+3,1-3,2-4$, the matrices $\psi(a),\psi(b)$
take the following form.
\begin{displaymath}
\begin{bmatrix}1&0&0\\0&0&1\\0&-1&0\end{bmatrix},\quad
\begin{bmatrix}1&0&0\\0&-1&0\\0&-1&0\end{bmatrix}
\end{displaymath}
Thus the syntactic representation of $X$ is the sum of of
two irreducible representations of dimensions $1$ and $2$.
\end{example}
\section{Cyclic sets}\label{sectionCyclic}
In the first part of
this section, we recall the definition of a cyclic set which was
introduced in~\cite{BerstelReutenauer1990}. In the second part, we
give a new proof of their complete reducibility.
In the last part, we connect the notion of cyclic sets with
that of  monoid characters.
\subsection{Cyclic and strongly cyclic sets}
A  subset 
$X$ of a monoid $M$ is \emph{cyclic} if it satisfies the two
following conditions.
\begin{enumerate}
\item[(i)] For any $u,v\in M$, one has $uv\in X$ if and only if $vu\in
  X$.
\item[(ii)] For any $w\in M$ and any integer $n\ge 1$, one has
$w^n\in X$ if and only if $w\in X$.
\end{enumerate}

If $\varphi$ is a morphism from a monoid $M$ onto a monoid $N$, for
any subset $X$ of $N$, the set $\varphi^{-1}(X)$ is cyclic if and only
if
$X$ is cyclic.

\begin{example}
The cyclic subsets of $a^*$ are the sets $\emptyset$, $1$, $a^+$ and $a^*$.
\end{example}

A rational set of words $X$  is \emph{strongly cyclic} if there is a morphism
$\varphi$ from $A^*$ into
 a finite monoid $M$ which has a zero
such that $X=\{x\in M\mid 0\notin \varphi(x^*)\}$.
Let $\A$ be a deterministic automaton with a set $Q$ of states.
The set of cyclically nonzero words defined by $\A$ is the set
\begin{equation}
X=\{x\in A^*\mid Q\cdot x^n\ne\emptyset \text{ for all $n\ge 0$}\}
.\label{eqStronglyCyclic}
\end{equation}
Note that since $Q$ is finite, for any $x\in X$ there is a $q\in Q$
such that $q\cdot x^n\ne 0$ for all $n\ge 0$.
\begin{proposition}
A set of words $X$ is strongly cyclic if and only if it
is the set of cyclically nonzero words defined by a
deterministic
automaton.
\end{proposition} 
\begin{proof}
The condition is necessary. Indeed, let $\varphi:A^*\rightarrow
M$
be a morphism into a finite monoid $M$ which has a zero such that
$X=\{x\in M\mid 0\notin \varphi(x^*)\}$. Let $\A$ be the automaton
with $M\setminus 0$ as set of states and with transitions defined by
$m\cdot a=m\varphi(a)$ if $m\varphi(a)\ne 0$. For any $x\in X$,
one has $1\cdot x^n\ne\emptyset$ for all $n\ge 0$. Thus $x$ satisfies
condition~\eqref{eqStronglyCyclic}. Conversely, if
$m\cdot x^n\ne\emptyset$ for some $m\in M\setminus 0$ and
for all $n\ge 0$, then
$0\notin \varphi(x^*)$. 

The condition is also sufficient. Indeed, assume that $X$ is the set
of cyclically nonzero
words defined by the deterministic automaton $\A$. Let $M$ be the transition
monoid
of  $\A$ and let $\varphi$ be the canonical morphism
from $A^*$ onto $M$. For any $x\in X$, one has $\varphi(x^n)\ne 0$
and thus $0\notin\varphi(x^*)$. Conversely, if $0\notin\varphi(x^*)$,
let $k$ be an integer such that $\varphi(x^k)$ is idempotent. Since
$\varphi(x^k)\ne 0$, there is a state $q$ such that $q\cdot
x^k\ne\emptyset$.
Then $q\cdot
x^{kn}\ne\emptyset$ for any $n\ge 0$ and consequently $q\cdot
x^{n}\ne\emptyset$ for any $n\ge 0$. Thus $X$ is strongly cyclic.
\end{proof}

 For a 
sequence of sets
$X_1,\ldots,X_n$ such that  $X_1\supset X_2\supset\ldots\supset X_n$,
the \emph{chain of differences} of the sequence is the set
 \begin{equation}
X=(X_1-X_2)+(X_3-X_4)+\ldots.\label{eqChains}
\end{equation}
The integer $n$ is called the length of the chain. According to the
parity
of $n$ the last term of the chain is $(X_{n-1}-X_n)$ or $(X_n)$. Note
that one can also write \eqref{eqChains} as $X=X_1-Y$
with $Y=(X_2-X_3)+(X_4-X_5)+\ldots$  a chain of differences of
length $n-1$ such that $Y\subset X$.

The following result is from~\cite{BealCartonReutenauer1996} (see the
proof of Theorem
10).
It shows in particular that any  cyclic rational set of words
 is a boolean combination of strongly
cyclic rational sets.
\begin{proposition}\label{propositionStrictlyCyclic}
Any  cyclic rational  set of words $X$ is a  chain of differences of
 strongly
cyclic rational sets. 
\end{proposition}

\begin{example}\label{exampleStictly}
Consider the automaton $\A$, called the \emph{even automaton}, 
represented in Figure~\ref{figureEven1} on the left. The automaton on
the right will be used below.
\begin{figure}[hbt]
\centering\gasset{Nadjust=wh}
\begin{picture}(50,15)(0,-3)
\node(1)(0,0){$1$}\node(2)(20,0){$2$}

\drawloop(1){$b$}\drawedge[curvedepth=5](1,2){$a$}
\drawedge[curvedepth=5](2,1){$a$}


\node(12)(40,0){$1,2$}

\drawloop(12){$-a$}
\end{picture}
\caption{The even automaton}\label{figureEven1}
\end{figure}
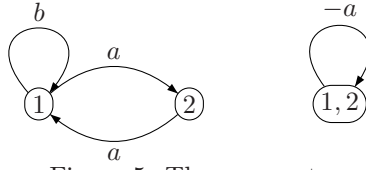
Let $X$ be the set of cyclically nonzero words for this automaton. We
have
\begin{displaymath}
X=a^*\cup (a^2)^*b\{aa,b\}^*\cup a(a^2)^*b\{aa,b\}^*a.
\end{displaymath}
The set $X$ is the union of two cyclic sets $a^*$ and 
$(a^2)^*b\{aa,b\}^*\cup a(a^2)^*b\{aa,b\}^*
a
$. The first one is
strongly cyclic but the second is not.
\end{example}
\subsection{Complete reducibility of cyclic sets}

The following result is from~\cite{BerstelReutenauer1990} (Corollary 12.2.2
\cite{BerstelReutenauer2011}). 

\begin{theorem}\label{theoremCyclic}
A cyclic rational  set of words is completely reducible.
\end{theorem}

A series $S$ is a \emph{trace series} if there exists a linear
representation $\mu$ of $A^*$ such that for any $w\in A^*$
\begin{displaymath}
(S,w)=\Trace(\mu w).
\end{displaymath}
The following is from~\cite{BerstelReutenauer1990} 
(it is Lemma 12.2.3 in~\cite{BerstelReutenauer2011}).
\begin{proposition}\label{propositionChristophe}
The syntactic algebra of a linear combination of trace
series is semisimple.
\end{proposition}

Let $\A$ be a finite deterministic automaton with set of states
$Q=\{1,2,\ldots,n\}$ on the alphabet $A$. Following~\cite{Beal1995},
 for $1\le k\le n$, the \emph{external power} of
order $k$ of $\A$ is the  weighted automaton $\A_k$  defined as
follows.
Its set of states is the set
$Q_k$ of sequences of integers $(i_1,i_2,\ldots,i_k)$ such that $1\le
i_1<i_2<\cdots <i_k\le n$. The edges are labeled in $A\cup -A$.
There is a transition by  $\varepsilon a$
from $(i_1,i_2,\ldots,i_k)$ to $(j_1,j_2,\ldots,j_k)$ if and only if
$(j_1,j_2,\ldots,j_k)$ is obtained from $(i_1\cdot a,i_2\cdot
a,\ldots,i_k\cdot a)$ by a permutation of signature $\varepsilon$.
\begin{example}
Let $\A$ be the even automaton of Example~\ref{exampleStictly}. The
external power $\A_2$ is represented in Figure~\ref{figureEven1} on the right.
\end{example}

The following combinatorial lemma on permutations is Lemma 6.4.9 
in~\cite{LindMarcus1995}.
\begin{lemma}\label{lemmaLindMarcus}
Let $\pi$ be a permutation of a finite set $P$ and let 
${\GR}=\{R\subset P\mid R\ne\emptyset,\pi(R)=R\}$. Then
\begin{displaymath}
\sum_{R\in \GR}(-1)^{\Card(R)+1}\varepsilon(\pi,R)=1
\end{displaymath}
where $\varepsilon(\pi,R)$ denotes the signature of the restriction of
$\pi$ to the set $R$.
\end{lemma}
We use Lemma~\ref{lemmaLindMarcus} to prove the following result.
\begin{proposition}\label{propStronglyTrace}
If $X$ is a strongly cyclic rational set, the series $\u(X)$ is a linear combination
of trace series.
\end{proposition}
\begin{proof}
Let $\A$ be a deterministic
automaton on the set $Q=\{1,2,\ldots n\}$
such that $X$ is the set of cyclically nonzero words defined by $\A$.
Denote by $\A_i$ for $1\le i\le n$ the external power of $\A$
of order $k$.
We denote by $\Trace_i(w)$ the trace of a word $w$ with respect
to the automaton $\A_i$. We have
\begin{equation}
\Trace_i(w)=\sum_{q\in Q_{i,w}} \varepsilon_{q,w}\label{eqTracei}
\end{equation}
where $Q_{i,w}$ is the set of $q\in Q_i$ such that $q\cdot w$ differs
from $q$ by a permutation of signature $\varepsilon_{q,w}$.

We claim that for each $x\in A^*$
\begin{equation}
(\u(X),x)=\sum_{i=1}^n(-1)^{i+1}\Trace_i(x).\label{eqTrace}
\end{equation}
This will imply the result by Proposition~\ref{propositionChristophe}.

To prove \eqref{eqTrace}, assume first that $x\in X$. Let $P\subset Q$ be the
largest set such that $x$ defines a permutation $\pi$ of $P$. Since $x$
is cyclically nonzero, $P$ is not empty. For each
$i=1,\ldots n$, $\Trace_i(x)=\sum_{q\in Q_{i,x}} \varepsilon_{q,x}$ by
Equation
\eqref{eqTracei}. But the set $Q_{i,x}$ is the set
of sequences $q=(q_1,\ldots,q_i)$ with
$q_1<\ldots<q_i$ such that the set $R=\{q_1,\ldots,q_i\}$ satisfies
$\pi(R)=R$.  These sequences are thus in bijection with the sets $R$
in $\GR=\{R\subset P\mid R\ne\emptyset,\pi(R)=R\}$. Thus
\begin{displaymath}
\sum_{i=1}^n(-1)^{i+1}\Trace_i(x)=\sum_{R\in\GR}(-1)^{\Card(R)+1}\epsilon(\pi,R)
\end{displaymath}
By Lemma~\ref{lemmaLindMarcus}
 the value of the right hand side is $1$. Thus we have
proved \eqref{eqTrace} for $x\in X$.

Next if $x\notin X$, then $\Trace_i(x)=0$ for all $i=1,\ldots n$.
Indeed, if  $\Trace_i(x)\ne 0$, there is a sequence $q_1,\ldots,q_i$
such that $q_1\cdot x=q_1,\ldots,q_i\cdot x=q_i$ and thus $x\in X$.
Thus the right handside of \eqref{eqTrace} is zero. This proves
\eqref{eqTrace} for $x\notin X$.
\end{proof}
We now give the proof of  Theorem~\ref{theoremCyclic}.
\begin{proof}
By Proposition~\ref{propositionStrictlyCyclic}, any cyclic rational set
is a chain of differences of strongly cyclic rational sets. We prove by induction
on the length $n$ of the chain that for any cyclic rational set $X$,
the series $\u(X)$ is a linear combination of trace series. By
Proposition~\ref{propositionChristophe} it implies the conclusion.

It is true when $n=0$ since then $X$ is empty.

Assume now that $n\ge 1$. Then $X=Y-Z$ where $Y$ is a strongly cyclic
rational set
and $Z\subset Y$ is a  chain of differences of length $n-1$ 
of strongly cyclic rational sets. by Proposition~\ref{propStronglyTrace}
$\u(Y)$ is a linear combination of trace series.
By induction hypothesis, $\u(Z)$ is
a linear combination of trace series. Since $\u(X)=\u(Y)-\u(Z)$,
 the same conclusion holds
for $\u(X)$.
\end{proof}

\begin{example}
Consider again the even automaton $\A$ represented in 
Figure~\ref{figureEven1} on the left. Let $X$ be the
set of cyclically nonzero words for $\A$.

 The minimal
automaton of $X$ is represented in Figure~\ref{figureMinimal}
\begin{figure}[hbt]
\centering\gasset{Nadjust=wh}
\begin{picture}(60,35)(0,-3)
\node[Nmarks=if,fangle=180](1)(0,0){$1$}
\node[Nmarks=f,fangle=180](2)(0,20){$2$}
\node[Nmarks=f](3)(20,0){$3$}\node(4)(20,20){$4$}
\node(5)(40,0){$5$}\node[Nmarks=f](6)(40,20){$6$}

\drawedge[curvedepth=5](1,2){$a$}\drawedge[curvedepth=5](2,1){$a$}
\drawedge(1,3){$b$}\drawedge(2,4){$b$}
\drawloop(4){$b$}\drawloop(3){$b$}
\drawedge[curvedepth=5](4,6){$a$}\drawedge[curvedepth=5](6,4){$a$}
\drawedge[curvedepth=5](3,5){$a$}\drawedge[curvedepth=5](5,3){$a$}
\end{picture}
\caption{The minimal automaton of $X$.}\label{figureMinimal}
\end{figure}
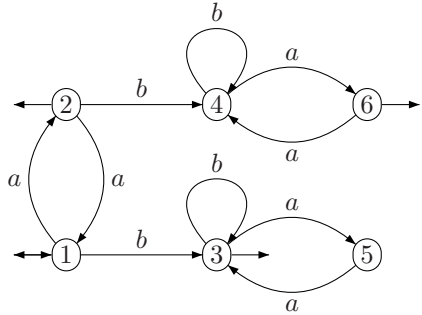

To obtain the
syntactic representation of $\u(X)$ we write
the linear representation associated with the automaton $\A$
 in the basis formed of the row vectors
\begin{displaymath}
1-3-6,\quad 2-4-5,\quad 3,\quad 5,\quad 4,\quad 6
\end{displaymath}
\begin{displaymath}
\psi(a)=\begin{bmatrix}0&1&0&0&0&0\\1&0&0&0&0&0\\0&0&0&1&0&0\\0&0&1&0&0&0\\
0&0&0&0&0&1\\0&0&0&0&1&0
\end{bmatrix},\quad
\psi(b)=\begin{bmatrix}0&0&0&0&0&0\\0&0&0&0&0&0\\0&0&1&0&0&0\\0&0&0&0&0&0\\
0&0&0&0&1&0\\0&0&0&0&0&0
\end{bmatrix}
\end{displaymath}
The initial and terminal vectors in this basis are
\begin{displaymath}
\lambda=\begin{bmatrix}1&0&1&0&0&1\end{bmatrix},\quad
\gamma=\begin{bmatrix}-1&1&1&0&0&1\end{bmatrix}^t
\end{displaymath}
In this way, the representation is a direct sum of three representations
of degrees $2,2,2$.
The first one is equivalent to a representation of degree $1$.

Thus the syntactic representation is the direct sum of three representations
of dimensions $1,2,2$. The first one is the linear representation
associated
with $\A_2$. The two other ones are equal to the linear representation
associated with $\A$ in such a way that the pair recognizes the trace
of the associated matrices.

\end{example}

\subsection{Characters of monoids}\label{sectionCharacters}
Let $M$ be a monoid.  A \emph{character}
on $M$  is a map of the form $m\mapsto \Trace(\rho m)$ where
$\rho:M\rightarrow\End(V)$ is a linear representation of $M$ 
over a finite dimensional vector space $V$. 
The character is irreducible if the representation is irreducible.
Any character is a sum of irreducible characters.

If $\varphi:A^*\rightarrow M$ is a morphism and $\chi$ is a character,
then $\chi\varphi$ is a completely reducible series. Indeed, this
is true if $\chi$ is irreducible and the general case follows from
the fact that any linear combination of completely reducible series
is completely reducible (Proposition~\ref{propositionChristophe1}).

The following result is from~\cite{McAlister1972}. It is proved for
$K=\C$ but the proof works for an algebraically closed field $K$
\cite{Steinberg2013}.
For an element $m$ of a finite semigroup $M$, we denote $m^\omega$
the idempotent of the semigroup generated by $m$.
\begin{theorem}\label{theoremMcAllistair}
Let $M$ be a finite monoid and let $K$ be an
algebraically closed field.
A map $f:M\rightarrow K$ is a linear combination of irreducible
characters if and only if
\begin{enumerate}
\item[\rm(i)] $f(xy)=f(yx)$ for any $x,y\in M$,
\item[\rm(ii)] $f(x^\omega x)=f(x)$ for any $x\in M$.
\end{enumerate}
\end{theorem}
This result gives an easy proof of theorem~\ref{theoremCyclic} 
over an algebraically closed field. Indeed, assume that $X$ is a  cyclic rational set
with syntactic morphism $\varphi:A^*\rightarrow M$. Let
$P=\varphi(X)$.
Then the characteristic function of $P$ satisfies the conditions
of Theorem~\ref{theoremMcAllistair}. This is clear for condition (i).
Next, $x^\omega x\in P$ implies that $x^n\in P$ for some $n\ge 1$
and thus implies $x\in P$. Conversely, if $x\in P$, then
$x^n\in P$ for all $n\ge 1$ and thus in particular
$x^\omega x\in P$. Thus condition (ii) is also
true.
Thus the characteristic function of $P$ is a linear combination
of irreducible characters. This implies that the characteristic
series of $X$ is a linear combination of trace series.

\bibliography{SyntacticAlgebra}
\bibliographystyle{plain}
\end{document}